\documentclass[abstract]{scrartcl}

\usepackage{amsmath}
\usepackage{amsfonts}
\usepackage{amssymb}
\usepackage{amsthm}
\usepackage{graphicx}
\usepackage[english]{babel}

\usepackage[labelfont=bf,textfont=it]{caption}
\setcapindent{0pt}

\usepackage{enumerate}
\usepackage{mathrsfs}

\renewcommand{\d}{\mathrm{d}}
\newcommand{\cA}{\mathcal{A}}
\newcommand{\cL}{\mathcal{L}}
\newcommand{\cV}{\mathscr{V}}

\newcommand{\tint}{{\textstyle \int}}
\renewcommand{\L}[2]{\frac{\partial L_{#1}}{\partial u_{#2}}}
\newcommand{\var}[3][]{\frac{\delta_{#1} {#2}}{\delta {#3}}}
\newcommand{\der}[2]{\frac{\partial {#1}}{\partial {#2}}}

\DeclareMathOperator{\pr}{pr}
\DeclareMathOperator{\D}{D}
\let\i\relax
\DeclareMathOperator{\i}{\iota}

\newtheorem{thm}{Theorem}
\newtheorem{lemma}[thm]{Lemma}
\newtheorem{prop}[thm]{Proposition}
\newtheorem{cor}[thm]{Corollary}
\theoremstyle{definition}
\newtheorem{defi}[thm]{Definition}
\theoremstyle{remark}
\newtheorem*{remark}{Remark}

\title{On the Lagrangian Structure \\ of Integrable Hierarchies}
\author{Yuri B. Suris, Mats Vermeeren \bigskip\\  Institut f\"ur Mathematik, MA 7-2, TU Berlin, \\ Str.\@ des 17.\@ Juni 136, 10623 Berlin, GERMANY}

\begin{document}

\maketitle

\begin{abstract}
We develop the concept of pluri-Lagrangian structures for integrable hierarchies. This is a continuous counterpart of the pluri-Lagrangian (or Lagrangian multiform) theory of integrable lattice systems. We derive the multi-time Euler Lagrange equations in their full generality for hierarchies of two-dimensional systems, and construct a pluri-Lagrangian formulation of the potential Korteweg-de Vries hierarchy.
\end{abstract}

\section{Introduction}

In this paper, our departure point are two developments which have taken place in the field of discrete integrable systems in recent years. 
\begin{itemize}

\item Firstly, multi-dimensional consistency of lattice systems has been proposed as a notion of integrability \cite{BS1, N}. In retrospect, this notion can be seen as a discrete counterpart of the well-known fact that integrable systems never appear alone but are organized into integrable \emph{hierarchies}. Based on the notion of multi-dimensional consistency, a classification of two-dimensional integrable lattice systems (the so called ABS list) was given in \cite{ABS}. Moreover, for all equations of the ABS list, considered as equations on $\mathbb Z^2$, a variational interpretation was found in \cite{ABS}. 

\item Secondly, the idea of the multi-dimensional consistency was blended with the variational formulation in \cite{LN1}, where it was shown that solutions of any ABS equation on any quad surface $\Sigma$ in $\mathbb{Z}^N$ are critical points of a certain action functional $\int_\Sigma\cL$ obtained by integration of a suitable discrete Lagrangian two-form $\cL$. Moreover, it was observed in \cite{LN1} that the critical value of the action remains invariant under local changes of the underlying quad-surface, or, in other words, that the 2-form $\cL$ is closed on solutions of quad-equations, and it was suggested to consider this as a defining feature of integrability. However, later research \cite{BPS} revealed that $\cL$ is closed not only on solutions of (non-variational) quad-equations, but also on general solutions of the corresponding Euler-Lagrange equations. Therefore, at least for discrete systems, the closedness condition is implicitly contained in the variational formulation.

\end{itemize}

A general theory of multi-time one-dimensional Lagrangian systems, both discrete and continuous, has been developed in \cite{S}. A first attempt to formulate the theory for continuous two-dimensional systems was made in \cite{S2}. For such systems, a solution is a critical point of the action functional $\int_S \cL$ on any two-dimensional surface $S$ in $\mathbb{R}^N$, where $\cL$ is a suitable differential two-form. The treatment in \cite{S2} was restricted to second order Lagrangians, i.e.\@ to two-forms $\cL$ that only depend on the second jet bundle. In the present work we will extend this to Lagrangians of any order.

As argued in \cite{BPS}, the unconventional idea to consider the action on arbitrary two-dimensional surfaces in the multi-dimensional space of independent variables has significant precursors. These include:
\begin{itemize}
\item {\em Theory of pluriharmonic functions} and, more generally, of pluriharmonic maps \cite{R, OV, BFPP}. By definition, a pluriharmonic function of several complex variables $f:\mathbb{C}^N\to\mathbb{R}$ minimizes the Dirichlet functional $E_\Gamma=\int_\Gamma |(f\circ \Gamma)_z|^2dz\wedge d\bar z$ along any holomorphic curve in its domain $\Gamma:\mathbb{C}\to\mathbb{C}^N$. Differential equations governing pluriharmonic functions,
\[
\frac{\partial^2 f}{\partial z_i\partial \bar{z}_j}=0 \quad {\rm for\;\;all}\quad i,j=1,\ldots,N,
\]
are heavily overdetermined. Therefore it is not surprising that pluriharmonic functions (and maps) belong to the theory of integrable systems. 

This motivates the term \emph{pluri-Lagrangian systems}, which was proposed in \cite{BPS, BS2}.

\item \emph{Baxter's Z-invariance} of solvable models of statistical mechanics \cite{Bax1, Bax2}. This concept is based on invariance of the partition functions of solvable models under elementary local transformations of the underlying planar graphs. It is well known (see, e.g., \cite{BoMeSu}) that one can identify planar graphs underlying these models with quad-surfaces in $\mathbb{Z}^N$. On the other hand, the classical mechanical analogue of the partition function is the action functional. This suggests the relation of Z-invariance to the concept of closedness of the Lagrangian 2-form, at least at the heuristic level. This relation has been made mathematically precise for a number of models, through the quasiclassical limit \cite{BMS1, BMS2}.

\item The classical notion of \emph{variational symmetry}, going back to the seminal work of E.~Noether \cite{Noether}, has been shown to be directly related to the closedness of the Lagrangian form in the multi-time \cite{S2}.
\end{itemize}

The main goal of this paper is two-fold: to derive the Euler Lagrange equations for two-dimensional pluri-Lagrangian problems of arbitrary order, and to state the (potential) KdV hierarchy as a pluri-Lagrangian system. We will also discuss the closedness of the Lagrangian two-form, which turns out to be related to the the Hamiltonian theory of integrable hierarchies.

Note that the influential monograph \cite{Dickey}, according to the foreword, is ``about hierarchies of integrable equations rather than about individual equations''. However, its Lagrangian part (chapters 19, 20) only deals with individual equations. The reason for this is apparently the absence of the concept of pluri-Lagrangian systems. We hope that this paper opens up the way for a variational approach to integrable hierarchies.

\section{Pluri-Lagrangian systems}\label{sect-pluriL}

\subsection{Definition}

We place our discussion in the formalism of the variational bicomplex as presented in \cite[Chapter 19]{Dickey} (and summarized, for the reader's convenience, in Appendix \ref{appendix-bicomp}). Slightly different versions of this theory can be found in \cite{Olver} and in \cite{Anderson}.

Consider a vector bundle $X: \mathbb{R}^N \rightarrow \mathbb{R}$ and its $n$-th jet bundle $J^n X$.
Let $\cL\in \cA^{(0,d)}(J^n X)$ be a smooth horizontal $d$-form. In other words, $\cL$ is a $d$-form on $\mathbb R^N$ whose coefficients depend on a function $u: \mathbb{R}^N \rightarrow \mathbb{R}$ and its partial derivatives up to order $n$. We call $\mathbb{R}^N$ the \emph{multi-time}, $u$ the \emph{field}, and $\cL$ the \emph{Lagrangian $d$-form}. We will use coordinates $(t_1, \ldots, t_N)$ on  $\mathbb{R}^N$. Recall that in the standard calculus of variations the Lagrangian is a \emph{volume form}, so that $d=N$.

\begin{defi}
We say that the field $u$ solves the \emph{pluri-Lagrangian problem} for $\mathcal{L}$ if $u$ is a critical point of the action $\int_S \mathcal{L}$ simultaneously for all $d$-dimensional surfaces $S$ in $\mathbb{R}^N$. The equations describing this condition are called the \emph{multi-time Euler-Lagrange equations}. We say that they form a \emph{pluri-Lagrangian system} and that $\mathcal{L}$ is a \emph{pluri-Lagrangian structure} for for these equations.
\end{defi}

To discuss critical points of a pluri-Lagrangian problem, consider the \emph{vertical derivative} $\delta \cL$ of the (0,$d$)-form $\cL$ in the variational bicomplex, and a \emph{variation} $\cV$. Note that we consider variations $\cV$ as vertical vector fields; such a restriction is justified by our interest, in the present paper, in autonomous systems only.  Besides, in the context of discrete systems only vertical vector fields seem to possess a natural analogs. The criticality condition of the action, $\delta \int_S \cL=0$, is described by the equation
\begin{equation}\label{variation}
\int_S \i_{\pr \cV} \delta \cL = 0,
\end{equation}
which has to be satisfied for any variation $\cV$ on $S$ that vanishes at the boundary $\partial S$. Recall that $\pr \cV$ is the $n$-th jet prolongation of the vertical vector field $\cV$, and that $\i$ stands for the contraction. One fundamental property of critical points can be established right from the outset.

\begin{prop}\label{prop-dLconst}
The exterior derivative $\d \cL$ of the Lagrangian is constant on critical points $u$.
\end{prop}

\begin{proof}

Consider a critical point $u$ and a small $(d+1)$-dimensional ball $B$. Because $S := \partial B$ has no boundary, Equation \eqref{variation} is satisfied for any variation $\cV$. Using Stokes' theorem and the properties that $\delta \d + \d \delta = 0$ and $\i_{\pr \cV} \d + \d\i_{\pr \cV} = 0$ (Propositions \ref{prop-delta-d} and \ref{prop-delta-i} in Appendix \ref{appendix-bicomp}), and , we find that
$$
0 = \int_{\partial B} \i_{\pr \cV} \delta \cL = \int_B \d (\i_{\pr \cV} \delta \cL) = -\int_B \i_{\pr \cV} \d(\delta \cL) = \int_B \i_{\pr \cV} \delta (\d \cL).
$$
Since this holds for any ball $B$ it follows that $\i_{\pr \cV} \delta (\d \cL) = 0$ for any variation $\cV$ of a critical point $u$. Therefore, $\delta(\d\cL)=0$, so that $\d \cL$ is constant on critical points $u$. Note that here we silently assume that the space of critical points is connected. It would be difficult to justify this property in any generality, but it is usually clear in applications, where the critical points are solutions of certain well-posed systems of partial differential equations.
\end{proof}

We will take a closer look at the property $\d \cL={\rm const}$ in Section \ref{sect-Hamiltonian}, when we discuss the link with Hamiltonian theory. It will be shown that vanishing of this constant, i.e., closedness of $\cL$ on critical points, is related to integrability of the multi-time Euler-Lagrange equations.

\subsection{Approximation by stepped surfaces}

For computations, we will use the multi-index notation for partial derivatives. For any multi-index $I=(i_1,\ldots,i_N)$ we set
$$
u_I = \frac{\partial^{|I|} u}{(\partial t_1)^{i_1} \ldots (\partial t_N)^{i_N}},
$$
where $|I| = i_1 + \ldots + i_N$. The notations $Ik$ and $Ik^\alpha$ will represent the multi-indices $(i_1,\ldots,i_k + 1, \ldots i_N)$ and $(i_1,\ldots,i_k + \alpha, \ldots i_N)$ respectively. When convenient we will also use the notations $It_k$ and $It_k^\alpha$ for these multi-indices. We will write $k \not\in I$ if $i_k = 0$ and $k \in I$ if $i_k > 0$. We will denote by $\D_i$ or $\D_{t_i}$ the total derivative with respect to coordinate direction $t_i$,
$$
\D_i := \D_{t_i} := \sum_I u_{I i} \der{}{u_I}
$$
and by $\D_I := \D_{t_1}^{i_1} \ldots \D_{t_1}^{i_N}$ the corresponding higher order derivatives.

Our main general result is the derivation of the multi-time Euler-Lagrange equations for two-dimensional surfaces ($d = 2$). That will allow us to study the KdV hierarchy as a pluri-Lagrangian system. However, it is instructive to first derive the multi-time Euler-Lagrange equations for curves ($d = 1$). 
\medskip

The key technical result used to derive multi-time Euler-Lagrange equations is the observation that it suffices to consider a very specific type of surface.

\begin{defi} 
A \emph{stepped $d$-surface} is a $d$-surface that is a finite union of coordinate $d$-surfaces. A \emph{coordinate $d$-surface} of the direction $(i_1, \ldots, i_d)$ is a $d$-surface lying in an affine $d$-plane $\{(t_1,\ldots,t_N) \mid  t_j = c_j\; \; {\rm for} \;\; j \neq  i_1,\ldots,i_d\}$. \end{defi}

\begin{lemma}\label{lemma-conv}
If the action is stationary on any stepped surface, then it is stationary on any smooth surface.
\end{lemma}

The proof of this Lemma can be found in appendix \ref{Appendix proof of lemma}.

\subsection{Multi-time Euler-Lagrange equations for curves}

\begin{thm}\label{thm-pEL-1}
Consider a Lagrangian 1-form $\mathcal{L} = \sum_{i=1}^N L_i \,\d t_i$. The multi-time Euler-Lagrange equations for curves are:
\begin{align}
&\var[i]{L_i}{u_I} = 0 \qquad \forall I \not\ni i, \label{pEL-1-0} \\
&\var[i]{L_i}{u_{Ii}} = \var[j]{L_j}{u_{Ij}}  \qquad \forall I, \label{pEL-1-1}
\end{align}
where $i$ and $j$ are distinct, 
and the following notation is used for the variational derivative corresponding to the coordinate direction $i$:
$$
\var[i]{L_i}{u_{I}} = \sum_{\alpha \ge 0} (-1)^\alpha \D_i^\alpha \der{L_i}{u_{I i^\alpha}} = \der{L_i}{u_I} - \D_i \der{L_i}{u_{I i}} + \D_i^2  \der{L_i}{u_{I i^2}} - \ldots.
$$
\end{thm}


\begin{remark}
In the special case that $\mathcal{L}$ only depends on the first jet bundle, system \eqref{pEL-1-0}--\eqref{pEL-1-1} reduces to the equations found in \cite{S}:
\begin{align*}
& \var[i]{L_i}{u} = 0 \quad  \Leftrightarrow\quad \frac{\partial L_i}{\partial u}- \D_i \frac{\partial L_i}{\partial u_i} =0, \\
& \var[i]{L_i}{u_j} = 0 \quad  \Leftrightarrow\quad \frac{\partial L_i}{\partial u_j}=0\quad {\rm for}\quad  i\neq j,\\
&\var[i]{L_i}{u_{i}} = \var[j]{L_j}{u_{j}} \quad  \Leftrightarrow\quad \frac{\partial L_i}{\partial u_i}=\frac{\partial L_j}{\partial u_j} \quad {\rm for}\quad i\neq j.
\end{align*}
\end{remark}

\begin{proof}[Proof of Theorem \ref{thm-pEL-1}]
According to Lemma \ref{lemma-conv}, it is sufficient to look at a general L-shaped curve $S = S_i \cup S_j$, where $S_i$ is a line segment of the coordinate direction $i$  and $S_j$ is a line segment of the coordinate direction $j$. Denote the cusp by $p := S_i \cap S_j$. We orient the curve such that $S_i$ induces the positive orientation on the point $p$ and $S_j$  the negative orientation. There are four cases, depending on how the $L$-shape is rotated. They are depicted in Figure \ref{fig-L}. To each case we associate a pair $(\varepsilon_i, \varepsilon_j) \in \{-1,+1\}^2$, where the positive value is taken if the respective piece of curve is oriented in the coordinate direction, and negative if it is oriented opposite to the coordinate direction.

\begin{figure}[htb]
\centering
\includegraphics[width=\linewidth]{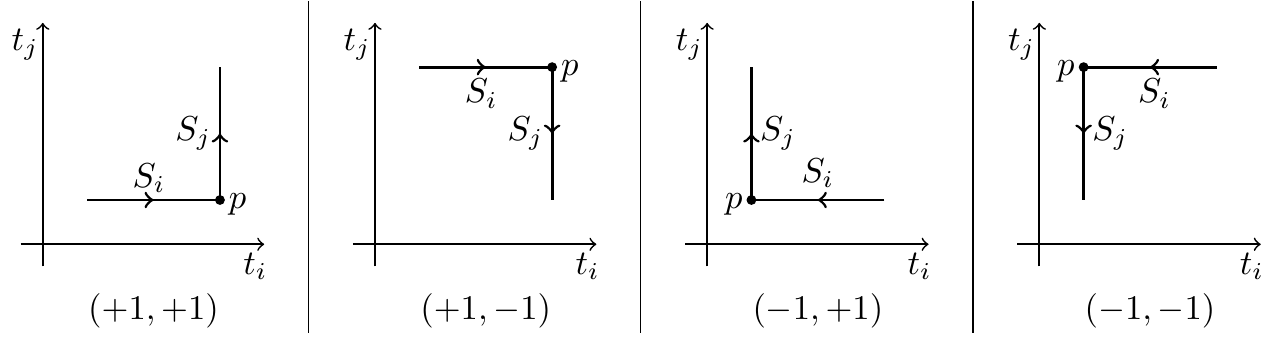}
\caption{The four L-shaped curves with their values of $(\varepsilon_i, \varepsilon_j)$.}
\label{fig-L}
\end{figure}

The variation of the action is
\begin{align*}
\int_{S} \i_{\pr \cV} \delta \cL
&= \varepsilon_i \int_{S_i} (\i_{\pr \cV} \delta L_i) \,\d t_i + \varepsilon_j \int_{S_j} (\i_{\pr \cV} \delta L_j) \,\d t_j \\
&= \varepsilon_i \int_{S_i} \sum_{I} \der{L_i}{u_I} \delta u_I(\cV) \,\d t_i + \varepsilon_j \int_{S_j} \sum_{I} \der{L_j}{u_I} \delta u_I(\cV) \,\d t_j.
\end{align*}
Note that these sums are actually finite. Indeed, since $\cL$ depends on the $n$-th jet bundle all terms with $|I| := i_1 + \ldots + i_N > n$ vanish.

Now we expand the sum in the first of the integrals and perform integration by parts.
\enlargethispage{1em}
\begin{align*}
& \varepsilon_i \int_{S_i} (\i_{\pr \cV} \delta L_i) \,\d t_i \\
&= \varepsilon_i \int_{S_i} \sum_{I \not\ni i} \left( \der{L_i}{u_I} \delta u_I(\cV) + \der{L_i}{u_{I i}} \delta u_{I i}(\cV) + \der{L_i}{u_{I i^2}} \delta u_{I i^2}(\cV) + \der{L_i}{u_{I i^3}} \delta u_{I i^3}(\cV) + \ldots \right) \d t_i\\
&= \varepsilon_i \int_{S_i} \sum_{I \not\ni  i} \left(\der{L_i}{u_I} - \D_i \der{L_i}{u_{I i}} + \D_i^2 \der{L_i}{u_{I i^2}} - \D_i^3 \der{L_i}{u_{I i^3}} + \ldots \right) \delta u_I(\cV) \, \d t_i\\
&\quad + \sum_{I \not\ni  i} \bigg( \der{L_i}{u_{I i}} \delta u_I(\cV) + \der{L_i}{u_{I i^2}} \delta u_{I i}(\cV) - \D_i \der{L_i}{u_{I i^2}} \delta u_{I}(\cV) \\
&\hspace{20mm} + \der{L_i}{u_{I i^3}} \delta u_{I i^2}(\cV) - \D_i \der{L_i}{u_{I i^3}}(\cV) \delta u_{I i}(\cV) + \D_i^2  \der{L_i}{u_{I i^3}} \delta u_{I}(\cV) + \ldots \bigg) \bigg|_p.
\end{align*}
Using the language of variational derivatives, this reads
\begin{align*}
\varepsilon_i \int_{S_i} (\i_{\pr \cV} \delta L_i) \,\d t_i &= \varepsilon_i \int_{S_i} \sum_{I \not\ni i} \var[i]{L_i}{u_I} \delta u_I(\cV) \,\d t_i \\
&\qquad + \sum_{I \not\ni i} \left( \var[i]{L_i}{u_{I i}} \delta u_I(\cV) + \var[i]{L_i}{u_{I i^2}} \delta u_{I i}(\cV) + \ldots \right) \bigg|_p\\
&= \varepsilon_i \int_{S_i} \sum_{I \not\ni i} \var[i]{L_i}{u_I} \delta u_I(\cV) \,\d t_i
+ \sum_{I}\bigg( \var[i]{L_i}{u_{I i}} \delta u_I(\cV) \bigg)\bigg|_p.
\end{align*}
The other piece, $S_j$, contributes
$$
\varepsilon_j \int_{S_j} \i_{\pr \cV} \delta L_j \,\d t_j = \varepsilon_j \int_{S_j} \sum_{I \not\ni j} \var[j]{L_j}{u_I} \delta u_I(\cV) \,\d t_j
- \sum_{I} \bigg( \var[j]{L_j}{u_{I j}} \delta u_I(\cV) \bigg)\bigg|_p,
$$
where the minus sign comes from the fact that $S_j$ induces negative orientation on the point $p$. Summing the two contributions, we find
\begin{align}
\int_{S} \i_{\pr \cV} \delta \cL  &= \varepsilon_i \int_{S_i} \sum_{I \not\ni i} \var[i]{L_i}{u_I} \delta u_I(\cV) \,\d t_i
+ \varepsilon_j \int_{S_j} \sum_{I \not\ni j} \var[j]{L_j}{u_I} \delta u_I(\cV) \,\d t_j \notag\\
&\qquad + \sum_{I} \bigg( \var[i]{L_i}{u_{I i}} \delta u_I(\cV) - \var[j]{L_j}{u_{I j}} \delta u_I(\cV) \bigg)\bigg|_p. \label{1dvariation}
\end{align}

Now require that the variation \eqref{1dvariation} of the action is zero for any variation $\cV$. If we consider variations that vanish on $S_j$, then we find for every multi-index $I$ which does not contain $i$ that
$$
\var[i]{L_i}{u_I} = 0.
$$
Given this equation, and its analogue for the index $j$, only the last term remains in the right hand side of Equation \eqref{1dvariation}. Considering variations around the cusp $p$ we find for every multi-index $I$ that
$$
\var[i]{L_i}{u_{Ii}} = \var[j]{L_j}{u_{Ij}}.
$$
It is clear these equations combined are also sufficient for the action to be critical.
\end{proof}

\subsection{Multi-time Euler-Lagrange equations for two-dimensional surfaces}

The two-dimensional case ($d = 2$) covers many known integrable hierarchies, including the potential KdV hierarchy which we will discuss in detail later on. We consider a Lagrangian two-form $\mathcal{L} = \sum_{i < j} L_{ij} \,\d t_i \wedge \d t_j$ and we will use the notational convention $L_{ji} = - L_{ij}$.

\begin{thm}\label{thm-pEL-2}
The multi-time Euler-Lagrange equations for two-dimensional surfaces are
\begin{align}
&\var[ij]{L_{ij}}{u_I} =0, & \forall I \not\ni i,j, \label{pEL-2-0}\\
&\var[ij]{L_{ij}}{u_{I j}} = \var[ik]{L_{ik}}{u_{I k}} &  \forall I \not\ni i, \label{pEL-2-1}\\
&\var[ij]{L_{ij}}{u_{I i j}} + \var[jk]{L_{jk}}{u_{I j k}} + \var[ki]{L_{ki}}{u_{I k i}} = 0 & \forall I, \label{pEL-2-2}
\end{align}
where $i$, $j$ and $k$ are distinct, 
and the following notation is used for the variational derivative corresponding to the coordinate directions $i,j$:
$$
\var[ij]{L_{ij}}{u_{I}} := \sum_{\alpha,\beta \ge 0} (-1)^{\alpha+\beta} \D_i^\alpha \D_j^\beta \der{L_{ij}}{u_{I i^\alpha j^\beta}}.
$$
\end{thm}

\begin{remark}
In the special case that $\mathcal{L}$ only depends on the second jet bundle, this system reduces to the equations stated in \cite{S2}.
\end{remark}

Before proceeding with the proof of Theorem \ref{thm-pEL-2}, we introduce some terminology and prove a lemma. A two-dimensional stepped surface consisting of $q$ flat pieces intersecting at some point $p$ is called a \emph{$q$-flower} around $p$, the flat pieces are called its \emph{petals}. If the action is stationary on every $q$-flower, it is stationary on any stepped surface. By Lemma \ref{lemma-conv} the action will then be stationary on any surface. The following Lemma shows that it is sufficient to consider 3-flowers.

\begin{lemma} If the action is stationary on every $3$-flower, then it is stationary on every $q$-flower for any $q > 3$.
\end{lemma}

\begin{proof} Let $F$ be a $q$-flower. Denote its petals corresponding to coordinate directions $(t_{i_1},t_{i_2})$, $(t_{i_2},t_{i_3})$, \ldots, $(t_{i_q},t_{i_1})$ by $S_{12}$, $S_{23}$, \ldots, $S_{q1}$ respectively. Consider the $3$-flower $F_{123} = S_{12} \cup S_{23} \cup S_{31}$, where $S_{31}$ is a petal in the coordinate direction $(t_{i_3},t_{i_1})$ such that $F_{123}$ is a flower around the same point as F. Similarly, define $F_{134}, \ldots, F_{1\,q-1\,q}$. Then (for any integrand)
\begin{align*}
\int_{F_{123}} + \int_{F_{134}} &+ \ldots + \int_{F_{1\,q-1\,q}} \\
=& \int_{S_{12}} + \int_{S_{23}} + \int_{S_{31}} + \int_{S_{13}} + \int_{S_{34}} + \int_{S_{41}} + \ldots + \int_{S_{1\,q-1}} + \int_{S_{q-1\,q}} + \int_{S_{q1}}.
\end{align*}
Here, $S_{21}$, $S_{32}$, \ldots are the petals $S_{12}$, $S_{23}$, \ldots but with opposite orientation (see Figure \ref{fig-floristry}). Therefore all terms where the index of $S$ contains $1$ cancel, except for the first and last, leaving
\[
\int_{F_{123}} + \ldots + \int_{F_{1\,q-1\,q}}
= \int_{S_{12}} + \int_{S_{23}} + \int_{S_{34}} + \ldots + \int_{S_{q-1\,q}} + \int_{S_{q1}} = \int_F. 
\]
By assumption the action is stationary on every 3-flower, so
\[ 
\int_F \i_{\pr \cV} \delta \cL = \int_{F_{123}} \i_{\pr \cV} \delta \cL + \ldots + \int_{F_{1\,q-1\,q}} \i_{\pr \cV} \delta \cL = 0. \qedhere 
\]
\end{proof}

\begin{figure}[h]
\begin{minipage}{.48\textwidth}
\centering
\includegraphics[width=\linewidth]{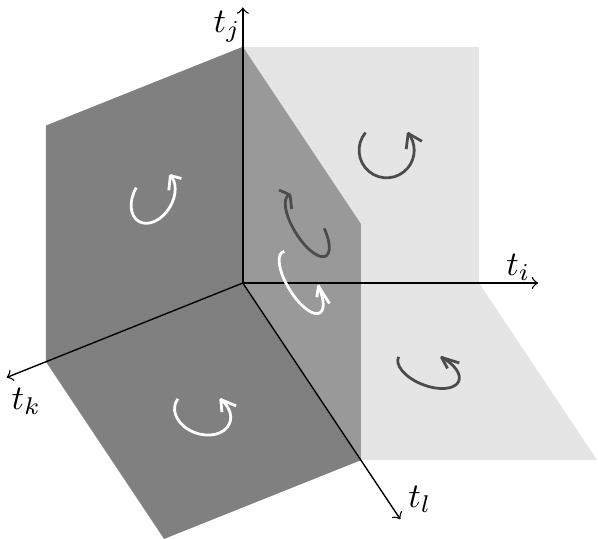}
\caption{Two 3-flowers composed to form a 4-flower. The common petal does not contribute to the integral because it occurs twice with opposite orientation.}
\label{fig-floristry}
\end{minipage}%
\hspace{.04\textwidth}%
\begin{minipage}{.48\textwidth}
\centering
\includegraphics[width=\linewidth]{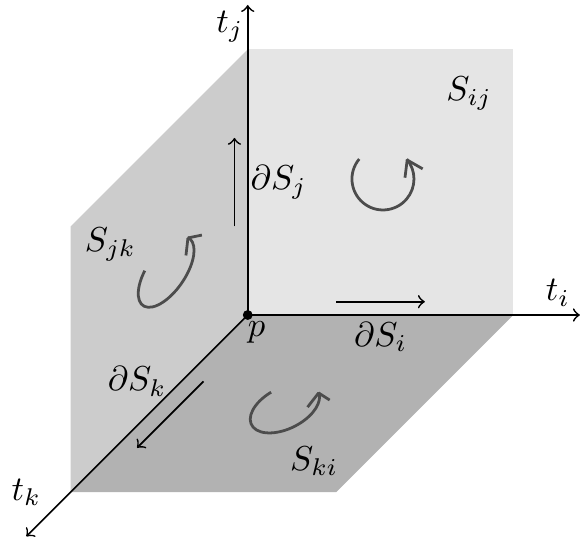}
\caption{A 3-flower. Different petals induce the opposite orientation on the common boundary.}
\label{fig-flower}
\end{minipage}
\end{figure}
\begin{figure}[t]
\centering
\includegraphics[width=\linewidth]{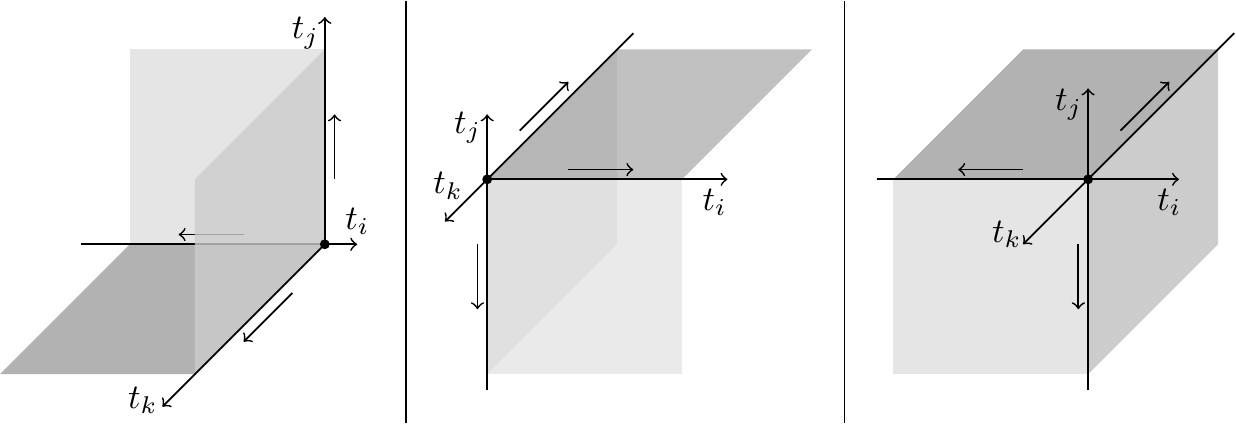}
\caption{Three of the other 3-flowers. The orientations of the interior edges do not all correspond to the coordinate direction.}
\label{fig-uglyflowers}
\end{figure}

\begin{proof}[Proof of Theorem \ref{thm-pEL-2}]
Consider a 3-flower $S = S_{ij} \cup S_{jk} \cup S_{ki}$ around the point $p = S_{ij} \cap S_{jk} \cap S_{ki}$. Denote its interior edges by
$$\partial S_i:= S_{ij} \cap S_{ki} ,\qquad \partial S_j := S_{jk} \cap S_{ij} ,\qquad \partial S_k := S_{ki} \cap S_{jk}.$$
On $S_i$, $S_j$ and $S_k$ we choose the orientations that induce negative orientation on $p$. We consider the case where these orientations correspond to the coordinate directions, as in Figure \ref{fig-flower}. The cases where one or more of these orientations are opposite to the corresponding coordinate direction (see Figure \ref{fig-uglyflowers}) can be treated analogously and yield the same result.

We choose the orientation on the petals in such a way that the orientations of $S_i$, $S_j$ and $S_k$ are induced by $S_{ij}$, $S_{jk}$ and $S_{ki}$ respectively. Then the orientations of $S_i$, $S_j$ and $S_k$ are the opposite of those induced by $S_{ki}$, $S_{ij}$ and $S_{jk}$ respectively (see Figure \ref{fig-flower}).

We will calculate 
\begin{equation}\label{actiononflower}
\int_{S} \i_{\pr \cV} \delta \mathcal{L} = \int_{S_{ij}} \i_{\pr \cV} \delta \mathcal{L} + \int_{S_{jk}} \i_{\pr \cV} \delta \mathcal{L} + \int_{S_{ki}} \i_{\pr \cV} \delta \mathcal{L}
\end{equation}
and require it to be zero for any variation $\cV$ which vanishes on the (outer) boundary of $S$. This will give us the multi-time Euler-Lagrange equations. 

For the first term of Equation \eqref{actiononflower} we find
\begin{align*}
\int_{S_{ij}} \i_{\pr \cV} \delta \mathcal{L} 
&=\int_{S_{ij}} \sum_I \L{ij}{I} \delta u_{I}(\cV) \, \d t_i \wedge \d t_j \\
&= \int_{S_{ij}} \sum_{I \not\ni i,j} \sum_{\lambda,\mu \ge 0} \L{ij}{I i^\lambda j^\mu} \delta u_{I i^\lambda j^\mu}(\cV) \, \d t_i \wedge \d t_j.
\end{align*}
First we perform integration by parts with respect to $t_i$ as many times as possible.
\begin{align*}
\int_{S_{ij}} \i_{\pr \cV} \delta \mathcal{L} 
&=\int_{S_{ij}} \sum_{I \not \ni i,j} \sum_{\lambda,\mu \ge 0} (-1)^\lambda \D_i^\lambda \L{ij}{I i^\lambda j^\mu} \delta u_{I j^\mu}(\cV) \, \d t_i \wedge \d t_j \\
&\qquad - \int_{\partial S_j} \sum_{I \not \ni i,j} \sum_{\lambda,\mu \ge 0} \sum_{\pi = 0}^{\lambda-1} (-1)^\pi \D_i^\pi \L{ij}{I i^\lambda j^\mu} \delta u_{I i^{\lambda - \pi -1}j^\mu}(\cV) \, \d t_j.
\end{align*}
Next integrate by parts with respect to $t_j$ as many times as possible.
\begin{align}
\int_{S_{ij}} \i_{\pr \cV} \delta \mathcal{L} 
\label{int-ij} &=\int_{S_{ij}} \sum_{I \not \ni i,j} \sum_{\lambda,\mu \ge 0} (-1)^{\lambda+\mu} \D_i^\lambda \D_j^\mu \L{ij}{I i^\lambda j^\mu} \delta u_{I}(\cV) \, \d t_i \wedge \d t_j \\
\label{bdy-j} &\qquad - \int_{\partial S_j} \sum_{I \not \ni i,j} \sum_{\lambda,\mu \ge 0} \sum_{\pi = 0}^{\lambda-1} (-1)^\pi \D_i^\pi \L{ij}{I i^\lambda j^\mu} \delta u_{I i^{\lambda - \pi -1}j^\mu}(\cV) \, \d t_j  \\
\label{bdy-i} &\qquad - \int_{\partial S_i} \sum_{I \not \ni i,j} \sum_{\lambda,\mu \ge 0} \sum_{\rho = 0}^{\mu-1} (-1)^{\lambda + \rho} \D_i^\lambda \D_j^\rho \L{ij}{I i^\lambda j^\mu} \delta u_{I j^{\mu - \rho -1}}(\cV) \, \d t_i.
\end{align}
The signs of \eqref{bdy-j} and \eqref{bdy-i} are due to the choice of orientations (see Figure \ref{fig-flower}). We can rewrite the integral $\eqref{int-ij}$ as
$$\int_{S_{ij}} \sum_{I \not \ni i,j} \var[ij]{L_{ij}}{u_I} \delta u_I(\cV) \, \d t_i \wedge \d t_j.$$
The last integral $\eqref{bdy-i}$ takes a similar form if we replace the index $\mu$ by $\beta = \mu - \rho - 1$.
\begin{align*}
-\int_{\partial S_i} \sum_{I \not \ni i,j} & \sum_{\lambda,\mu \ge 0} \sum_{\rho = 0}^{\mu-1} (-1)^{\lambda + \rho} \D_i^\lambda \D_j^\rho \L{ij}{I i^\lambda j^\mu} \delta u_{I j^{\mu - \rho -1}}(\cV) \, \d t_i \\
&= -\int_{\partial S_i} \sum_{I \not \ni i,j} \sum_{\beta,\lambda,\rho \ge 0} (-1)^{\lambda + \rho} \D_i^\lambda \D_j^\rho \L{ij}{I i^\lambda j^{\beta + \rho + 1}} \delta u_{I j^\beta}(\cV) \, \d t_i \\
&= -\int_{\partial S_i} \sum_{I \not \ni i,j} \sum_{\beta \ge 0} \var[ij]{L_{ij}}{u_{I j^{\beta+1}}} \delta u_{I j^\beta}(\cV) \,\d t_i.
\end{align*}
To write the other boundary integral $\eqref{bdy-j}$ in this form we first perform integration by parts.
\begin{align*}
- \int_{\partial S_j} \sum_{I \not \ni i,j} & \sum_{\lambda,\mu \ge 0} \sum_{\pi = 0}^{\lambda-1} (-1)^\pi \D_i^\pi \L{ij}{I i^\lambda j^\mu} \delta u_{I i^{\lambda - \pi -1} j^\mu}(\cV) \, \d t_j \\
&= - \int_{\partial S_j} \sum_{I \not \ni i,j} \sum_{\lambda,\mu \ge 0} \sum_{\pi = 0}^{\lambda-1} (-1)^{\pi + \mu} \D_i^\pi \D_j^\mu \L{ij}{I i^\lambda j^\mu} \delta u_{I i^{\lambda - \pi -1}}(\cV) \, \d t_j \\
&\qquad + \sum_{I \not \ni i,j} \sum_{\lambda,\mu \ge 0} \sum_{\pi = 0}^{\lambda-1} \sum_{\rho = 0}^{\mu-1} (-1)^{\pi + \rho} \left.\left( \D_i^\pi \D_j^\rho \L{ij}{I i^\lambda j^\mu} \delta u_{I i^{\lambda - \pi -1} j^{\mu - \rho -1}}(\cV) \right) \right|_p .
\end{align*}
Then we replace $\lambda$ by $\alpha = \lambda - \pi -1$ and in the last term $\mu$ by $\beta = \mu - \rho - 1$.
\begin{align*}
- \int_{\partial S_j} & \sum_{I \not \ni i,j} \sum_{\lambda,\ge 0} \sum_{\pi = 0}^{\lambda-1} (-1)^\pi \D_i^\pi \L{ij}{I i^\lambda j^\mu} \delta u_{I i^{\lambda - \pi -1} j^\mu}(\cV) \, \d t_j \\
&= - \int_{\partial S_j} \sum_{I \not \ni i,j} \sum_{\alpha,\mu,\pi \ge 0} (-1)^{\pi + \mu} \D_i^\pi \D_j^\mu \L{ij}{I i^{\alpha + \pi + 1} j^\mu} \delta u_{I i^\alpha}(\cV) \, \d t_j \\
&\qquad +  \sum_{I \not \ni i,j} \sum_{\alpha,\beta,\pi,\rho \ge 0} \left.\left( (-1)^{\pi + \rho} \D_i^\pi \D_j^\rho \L{ij}{I i^{\alpha + \pi + 1} j^{\beta + \rho + 1}} \delta u_{I i^\alpha j^\beta}(\cV) \right) \right|_p \\
&= - \int_{\partial S_j} \sum_{I \not \ni i,j} \sum_{\alpha \ge 0} \var[ij]{L_{ij}}{u_{I i^{\alpha + 1}}} \delta u_{I i^\alpha}(\cV) \, \d t_j  +   \sum_{I \not \ni i,j} \sum_{\alpha,\beta \ge 0} \left.\left( \var[ij]{L_{ij}}{u_{I i^{\alpha + 1} j^{\beta + 1}}} \delta u_{I i^\alpha j^\beta}(\cV) \right)\right|_p .
\end{align*}
Putting everything together we find
\begin{align*}
\int_{S_{ij}} \i_{\pr \cV} \delta \mathcal{L}  
&= \int_{S_{ij}} \sum_{I \not \ni i,j} \var[ij]{L_{ij}}{u_I} \delta u_I(\cV) \, \d t_i \wedge \d t_j
- \int_{\partial S_i} \sum_{I \not \ni i} \var[ij]{L_{ij}}{u_{Ij}} \delta u_{I}(\cV) \d t_i \\
&\qquad - \int_{\partial S_j} \sum_{I \not \ni j} \var[ij]{L_{ij}}{u_{Ii}} \delta u_{I}(\cV) \, \d t_j 
+  \left.\left( \sum_{I}  \var[ij]{L_{ij}}{u_{I i j}} \delta u_{I}(\cV) \right)\right|_p .
\end{align*}
Expressions for the integrals over $S_{jk}$ and $S_{ki}$ are found by cyclic permutation of the indices. Finally we obtain
\begin{align}
\int_{S} \i_{\pr \cV} \delta \mathcal{L} 
&= \int_{S_{ij}} \sum_{I \not \ni i,j} \var[ij]{L_{ij}}{u_I} \delta u_I(\cV) \, \d t_i \wedge \d t_j \notag \\
&\qquad - \int_{\partial S_i} \Bigg(\sum_{I \not \ni i}  \var[ij]{L_{ij}}{u_{I j}} \delta u_{I}(\cV) + \sum_{I \not \ni i} \var[ki]{L_{ki}}{u_{I k}} \delta u_{I}(\cV) \Bigg)  \, \d t_i \label{varaction} \\
&\qquad + \sum_{I} \left.\left( \var[ij]{L_{ij}}{u_{I i j}} \delta u_{I}(\cV) \right) \right|_p 
\qquad + \text{ cyclic permutations in } i,j,k. \notag
\end{align}
From this we can read off the multi-time Euler-Lagrange equations.
\end{proof}

\section{Pluri-Lagrangian structure of the Sine-Gordon equation}\label{sect-SG}
 
We borrow our first example of a pluri-Lagrangian system from \cite{S2}.

The Sine-Gordon equation $u_{xy}=\sin u$ is the Euler-Lagrange equation for
$$L = \frac{1}{2}u_xu_y-\cos u.$$
Consider the vector field $\varphi \der{}{u}$ with
$$\varphi=u_{xxx}+\frac{1}{2}u_x^3$$
and its prolongation 
$\D_\varphi := \sum_I \varphi_I \der{}{u_I}.$
It is known that $\D_\varphi$ is a variational symmetry for the sine-Gordon equation \cite[p. 336]{Olver}. In particular, we have that
\begin{equation}\label{varsym}
\D_\varphi L=  \D_x N + \D_y M
\end{equation}
with
\begin{align*}
M & = \dfrac{1}{2}\varphi u_x-\dfrac{1}{8}u_x^4+\dfrac{1}{2}u_{xx}^2,  \\
N & = \dfrac{1}{2}\varphi u_y -\dfrac{1}{2}u_x^2\cos u - u_{xx}(u_{xy}-\sin u). 
\end{align*}

Now we introduce a new independent variable $z$ corresponding to the ``flow'' of the generalized vector field $\D_\varphi$, i.e.\@ $u_z = \varphi$.
Consider simultaneous solutions of the Euler-Lagrange equation $\var{L}{u} = 0$ and of the flow $u_z = \varphi$ as functions of three independent variables $x,y,z$.  Then Equation \eqref{varsym} expresses the closedness of the two-form 
$$\cL =  L\, dx \wedge dy - M\, dz \wedge dx - N\, dy\wedge dz.$$
The fact that $\d \cL = 0$ on solutions is consistent with Proposition \ref{prop-dLconst}. Hence $\cL$ is a reasonable candidate for a Lagrangian two-form.

\begin{thm}
The multi-time Euler-Lagrange equations for the Lagrangian two-form 
$$\cL =  L_{12}\ dx\wedge dy + L_{13}\ dx\wedge dz + L_{23}\ dy\wedge dz$$
with the components 
\begin{align}
L_{12} & = \frac{1}{2}u_xu_y-\cos u, \label{eq: SG Lagr L}\\
L_{13} & = \frac{1}{2}u_xu_z-\dfrac{1}{8}u_x^4+\dfrac{1}{2}u_{xx}^2, \label{eq: SG Lagr M}\\
L_{23} & = -\frac{1}{2}u_yu_z+\dfrac{1}{2}u_x^2\cos u+u_{xx}(u_{xy}-\sin u), \label{eq: SG Lagr N}
\end{align}
consist of the sine-Gordon equation $$u_{xy} = \sin u,$$ the modified KdV equation $$u_z = u_{xxx} + \frac{1}{2} u_x^3,$$ and corollaries thereof. On solutions of either of these equations the two-form $\cL$ is closed. 
\end{thm}

\enlargethispage{1em}
\begin{proof}
Let us calculate the multi-time Euler-Lagrange equations \eqref{pEL-2-0}--\eqref{pEL-2-2} one by one:

\begin{itemize}
\item The equation $\displaystyle \var[12]{L_{12}}{u} = 0$ yields \dotfill  $u_{xy} = \sin u$.

For any $\alpha > 0$  the equation $\displaystyle \var[12]{L_{12}}{u_{z^\alpha}} = 0$ yields $0=0$.

\item The equation $\displaystyle \var[13]{L_{13}}{u} = 0$ yields \dotfill  $u_{xz} = \tfrac{3}{2} u_x^2 u_{xx} + u_{xxxx}$.

For any $\alpha > 0$ the equation $\displaystyle \var[13]{L_{13}}{u_{y^\alpha}} = 0$ yields $0=0$.

\item The equation $\displaystyle \var[23]{L_{23}}{u} = 0$ yields \dotfill  $u_{yz} = \tfrac{1}{2} u_x^2 \sin u + u_{xx} \cos u$.

The equation $\displaystyle \var[23]{L_{23}}{u_x} = 0$ yields \dotfill $u_{yxx} = u_x \cos u$.

The equation $\displaystyle \var[23]{L_{23}}{u_{xx}} = 0$ yields \dotfill $u_{xy} = \sin u$.

For any $\alpha > 2$ , the equation $\displaystyle \var[23]{L_{23}}{u_{x^\alpha}} = 0$ yields $0=0$.

\item The equation $\displaystyle \var[13]{L_{13}}{u_x} = \var[23]{L_{23}}{u_y}$ yields \dotfill  $u_z = u_{xxx} + \tfrac{1}{2} u_x^3$.

The equation $\displaystyle \var[13]{L_{13}}{u_{xx}} = \var[23]{L_{23}}{u_{xy}}$ yields $u_{xx} = u_{xx}$.

For any other $I$ the equation $\displaystyle \var[13]{L_{13}}{u_{Ix}} = \var[23]{L_{23}}{u_{Iy}}$ yields $0=0$.

\item The equation $\displaystyle \var[12]{L_{12}}{u_y} = \var[13]{L_{13}}{u_z}$ yields $\frac{1}{2}u_x = \frac{1}{2}u_x$.

For any nonempty $I$, the equation $\displaystyle \var[12]{L_{12}}{u_{Iy}} = \var[13]{L_{13}}{u_{Iz}}$ yields $0=0$.

\item The equation $\displaystyle \var[12]{L_{12}}{u_x} = \var[23]{L_{32}}{u_z}$ yields $\frac{1}{2}u_y = \frac{1}{2}u_y$.

For any nonempty $I$, the equation $\displaystyle \var[12]{L_{12}}{u_{Iy}} = \var[23]{L_{32}}{u_{Iz}}$ yields $0=0$.

\item For any $I$ the equation $\displaystyle\var[12]{L_{12}}{v_{Ixy}} + \var[23]{L_{23}}{v_{Iyz}} + \var[13]{L_{31}}{v_{Izx}} = 0$ yields $0=0$.
\end{itemize}
It remains to notice that all nontrivial equations in this list are corollaries of the equations $u_{xy} = \sin u$ and $u_z = u_{xxx} + \tfrac{1}{2} u_x^3$, derived by differentiation.

The closedness of $\cL$ can be verified by direct calculation:
\begin{align*}
\D_z L_{12} - \D_y L_{13} + \D_x L_{23}
&= \frac{1}{2} (u_{yz} u_x + u_{xz} u_y) + u_z \sin u \\
&\qquad - \frac{1}{2} u_{yz} u_x - \frac{1}{2} u_z u_{xy} + \frac{1}{2} u_x^3 u_{xy} - u_{xx} u_{xxy} \\
&\qquad - \frac{1}{2} u_{xz} u_y - \frac{1}{2} u_z u_{xy} + u_x u_{xx} \cos u - \frac{1}{2} u_x^3 \sin u \\
&\qquad + u_{xxx} (u_{xy} - \sin u) + u_{xx} (u_{xxy} - u_x \cos u) \\
&= - \left(u_z - \frac{1}{2} u_x^3 - u_{xxx}\right)(u_{xy} - \sin u). \qedhere
\end{align*}
\end{proof}

\begin{remark}
The Sine-Gordon equation and the modified KdV equation are the simplest equations of their respective hierarchies. Furthermore, those hierarchies can be seen as the positive and negative parts of one single hierarchy that is infinite in both directions \cite[sect. 3c and 5k]{Newell}. It seems likely that this whole hierarchy possesses a pluri-Lagrangian structure.
\end{remark}

\section{The KdV hierarchy}\label{sect-KdV}

Our second and the main example of a pluri-Lagrangian system will be the (potential) KdV hierarchy. This section gives an overview of the relevant known facts about KdV, mainly following Dickey \cite[Section 3.7]{Dickey}. The next section will present its pluri-Lagrangian structure.

One way to introduce the \emph{Korteweg-de Vries (KdV) hierarchy} is to consider a formal power series 
$$
R = \sum_{k=0}^\infty r_k z^{-2k-1},
$$
with the coefficients $r_k=r_k[u]$ being polynomials of $u$ and its partial derivatives with respect to $x$, satisfying the equation
\begin{equation}\label{resolvent}
R_{xxx} + 4uR_x + 2u_x R - z^2 R_x = 0.
\end{equation}
Multiplying this equation by $R$ and integrating with respect to $x$ we find
\begin{equation}\label{resolvent-int}
R R_{xx} - \tfrac{1}{2} R_x^2 + 2 \left( u - \tfrac{1}{4} z^2 \right) R^2 = C(z),
\end{equation}
where $C(z) = \sum_{k = 0}^\infty c_k z^{-2k}$ is a formal power series in $z^{-2}$, with coefficients $c_k$ being constants.  
%
%
Different choices of $C(z)$ correspond to different normalizations of the KdV hierarchy. We take $C(z) = \tfrac{1}{8}$, i.e.\@ $c_0 = \tfrac{1}{8}$ and $c_k = 0$ for $k >0$. The first few coefficients of the power series $R = r_0 z^{-1} + r_1 z^{-3} + r_2 z^{-5} + \ldots$ are 
$$
r_0 = \tfrac{1}{2}, \quad r_1 = u, \quad r_2 = u_{xx} + 3 u^2, \quad r_{3} = u_{xxxx} + 10 u u_{xx} + 5 u_x^2 + 10 u^3.
$$

The Korteweg-de Vries hierarchy is defined as follows.
\begin{defi}
\begin{itemize}
\item The \emph{KdV hierarchy} is the family of equations
$$
u_{t_k} = ( r_k[u] )_x.
$$
\item Write $g_k[v] := r_k[v_x]$. The \emph{potential KdV (PKdV) hierarchy} is the family of equations
$$
v_{t_k} = g_k[v] .
$$ 
\item The \emph{differentiated potential KdV (DPKdV) hierarchy} is the family of equations
$$
v_{x t_k} = ( g_k[v] )_x.$$ 
\end{itemize}
\end{defi}
The right-hand sides of first few PKdV equations are
$$
g_1 = v_x, \quad g_2 = v_{xxx} + 3 v_x^2, \quad g_{3} = v_{xxxxx} + 10 v_x v_{xxx} + 5 v_{xx}^2 + 10 v_x^3.
$$

\begin{remark}
The first KdV and PKdV equations, $u_{t_1} = u_x$, resp. $v_{t_1}=v_x$,  allow us to identify $x$ with $t_1$.
\end{remark}

\begin{prop}\label{prop-KdV-var}
The differential polynomials $r_k[u]$ satisfy
$$
\var{r_k}{u} = \left(4 k - 2\right) r_{k-1},
$$
where $\var{}{u}$ is shorthand notation for $\var[1]{}{u}$.
\end{prop}
A proof of this statement can be found in \cite[3.7.11--3.7.14]{Dickey}.

\begin{cor}\label{cor-KdV-var}
Set $h_k [v]:= \frac{1}{4 k + 2}g_{k+1}[v]$, then the differential polynomials $g_k$ and $h_k$ satisfy
$$\var{g_k}{v_x} = \left(4 k - 2\right) g_{k-1} \qquad \text{and} \qquad \var{h_k}{v_x} = g_k.$$
\end{cor}

Before we proceed, let us formulate a simple Lemma.

\begin{lemma}\label{lemma-var-xx}
For any multi-index $I$ and for any differential polynomial $f[v]$ we have:
$$ \D_{x} \left( \var{f}{v_{Ix}} \right) = \der{f}{v_I} - \var{f}{v_I}.$$
\end{lemma}
\begin{proof}
By direct calculation:
\begin{align*}
\D_{x} \left( \var{f}{v_x} \right) 
&= \D_x \left( \der{f}{v_{Ix}} - \D_x \der{f}{v_{Ix^2}} + \D_x^2 \der{f}{v_{Ix^3}} + \ldots \right) \\
&= \D_x \der{f}{v_{Ix}} - \D_x^2 \der{f}{v_{Ix^2}} + \D_x^2 \der{f}{v_{Ix^3}} + \ldots
=\der{f}{v_I} - \var{f}{v_I}. \qedhere
\end{align*}
\end{proof}

We can now  find Lagrangians for the the DPKdV equations.

\begin{prop} \label{thm Lagrangian PKdV}
The DPKdV equations are Lagrangian, with the Lagrange functions
$$
L_k[v] = \frac{1}{2}v_x v_{t_k} - h_k[v].
$$
\end{prop}
\begin{proof}
Since $h_k = \frac{1}{4k + 2}g_{k+1}$ does not depend on $v$ directly, it follows from Lemma \ref{lemma-var-xx} and Corollary \ref{cor-KdV-var} that\begin{align*}
& \var{L_k}{v} = -v_{t_k x} - \var{h_k}{v} = - v_{t_k x} + \D_x\var{h_k}{v_x} = -v_{t_k x} + (g_k)_x. \qedhere
\end{align*}
\end{proof}

\section{Pluri-Lagrangian structure of PKdV hierarchy}\label{sect-PKdV-EL}

Since the individual KdV and PKdV equations are evolutionary (not variational), it seems not very plausible that they could have a pluri-Lagrangian structure. However, it turns out that the PKdV hierarchy as a whole is pluri-Lagrangian. Let us stress that this structure is only visible if one considers several PKdV equations simultaneously and not individually. We consider a finite-dimensional multi-time $\mathbb{R}^N$ parametrized by $t_1,t_2,\ldots , t_N$ supporting the first $N$ flows of the PKdV hierarchy. Recall that the first PKdV equation reads $v_{t_1} = v_{x}$, which allows us to identify $t_1$ with $x$.

The formulation of the main result involves certain differential polynomials introduced in the following statement.
\begin{lemma}\label{lemma-ab}
\begin{itemize}
\item
There exist differential polynomials $b_{ij}[v]$ depending on $v$ and $v_{x^\alpha}$, $\alpha>0$, such that 
\begin{equation}\label{eq def b}
\D_x(g_i) g_j = \D_x(b_{ij}).
\end{equation}

\item These polynomials satisfy
\begin{equation}\label{lemma-b+b}
b_{ij} + b_{ji} = g_i g_j.
\end{equation}

\item The differential polynomials $a_{ij}[v]$ (depending on $v_{x^\alpha}$ and $v_{x^\alpha t_j}$, $\alpha\ge 0$) defined by
\begin{equation}\label{eq def a}
a_{ij} := v_{t_j} \var[1]{h_i}{v_x} + v_{x t_j} \var[1]{h_i}{v_{xx}} + v_{xx t_j} \var[1]{h_i}{v_{xxx}} + \ldots
\end{equation}
satisfy 
\begin{equation}\label{lemma-Da}
\D_j(h_i) + \D_{x}(g_i) v_{t_j} = \D_x(a_{ij}).
\end{equation}
\end{itemize}
\end{lemma}
\begin{proof}
The existence of polynomials $b_{ij}$  is shown in \cite[3.7.9]{Dickey}. Since
$$\D_x \left( b_{ij} + b_{ji} \right) = \D_x(g_i) g_j + g_i \D_x(g_j) = \D_x(g_i g_j),$$
and since neither $b_{ij} + b_{ji}$ nor $g_i g_j$ contain constant terms, Equation (\ref{lemma-b+b}) follows. The last claim is a straightforward calculation using Lemma \ref{lemma-var-xx}:
\begin{align*}
\D_x(a_{ij}) &= \D_x \left(v_{t_j} \var[1]{h_i}{v_x} + v_{x t_j} \var[1]{h_i}{v_{xx}} + v_{xx t_j} \var[1]{h_i}{v_{xxx}} + \ldots \right) \\
&= v_{x t_j} \var[1]{h_i}{v_x} + v_{xx t_j} \var[1]{h_i}{v_{xx}} + v_{xxx t_j} \var[1]{h_i}{v_{xxx}} + \ldots \\
&\quad + v_{t_j} \D_x \left(\var[1]{h_i}{v_x}\right) + v_{x t_j} \D_x \left(\var[1]{h_i}{v_{xx}}\right) + v_{xx t_j} \D_x \left(\var[1]{h_i}{v_{xxx}}\right)+ \ldots \\
&= v_{x t_j} \var[1]{h_i}{v_x} + v_{xx t_j} \var[1]{h_i}{v_{xx}} + v_{xxx t_j} \var[1]{h_i}{v_{xxx}} + \ldots \\
&\quad - v_{t_j} \var[1]{h_i}{v} + v_{t_j} \der{h_i}{v} - v_{x t_j} \var[1]{h_i}{v_{x}} + v_{x t_j} \der{h_i}{v_{x}} - v_{xx t_j} \var[1]{h_i}{v_{xx}} + v_{xx t_j} \der{h_i}{v_{xx}} - \ldots \\
&= \D_j h_i - v_{t_j} \var[1]{h_i}{v} = \D_j h_i + \D_x(g_i) v_{t_j} \qedhere
\end{align*}
\end{proof}

Now we are in a position to give a pluri-Lagrangian formulation of the PKdV hierarchy. 
\begin{thm}\label{thm-kdvEL}
The multi-time Euler-Lagrange equations for the Lagrangian two-form $\cL = \sum_{i<j} L_{ij} \,\d t_i \wedge \d t_j$, with coefficients given by 
\begin{equation} \label{L1i}
L_{1i} := L_i = \frac{1}{2} v_x v_{t_i} - h_i 
\end{equation}
and 
\begin{equation}\label{Lij}
L_{ij} := \tfrac{1}{2} (v_{t_i} g_j - v_{t_j} g_i) + (a_{ij} - a_{ji}) - \tfrac{1}{2}(b_{ij} - b_{ji}) \quad \text{for} \quad j>i>1
\end{equation} 
are the first $N-1$ nontrivial PKdV equations
$$v_{t_2} = g_2, \quad v_{t_3} = g_3, \quad \ldots \quad v_{t_N} = g_N,$$
and equations that follow from these by differentiation.
\end{thm}

\subsection{Variational symmetries and the pluri-Lagrangian form}

Before proving Theorem \ref{thm-kdvEL}, let us give an heuristic derivation of expression \eqref{Lij} for $L_{ij}$. 
The ansatz is that different flows of the PKdV hierarchy should be variational symmetries of each other. (We are grateful to V. Adler who proposed  this derivation to us in a private communication.)

Fix two distinct integers $i,j \in \{ 2, 3, \ldots, N \}$. Consider the the $i$-th DPKdV equation, which is nothing but the conventional two-dimensional variational system generated in the $(x,t_i)$-plane by the Lagrange function 
$$
L_{1i} [v]= \frac{1}{2} v_x v_{t_i} - h_i[v].
$$
Consider the evolutionary equation $v_{t_j} = g_j[v]$, i.e., the $j$-th PKdV equation, and the corresponding generalized vector field
$$
\D_{g_j} := \sum_{I\not\ni j} (\D_I g_j)\frac{\partial}{\partial v_I}.
$$
We want to show that $\D_{g_j}$ is a variational symmetry of $L_{1i}$. For this end, we look for $L_{ij}$ such that
\begin{equation}\label{varsym-ansatz}
\D_{g_j} (L_{1i}) - \D_i \left(L_{1j}^{(g_j)}\right) +  \D_x (L_{ij})= 0.
\end{equation} 
Here, $L_{1j}^{(g_j)}$ is the Lagrangian defined by (\ref{L1i}) but with $v_{t_j}$ replaced by $g_j$:
$$
L_{1j}^{(g_j)} := \frac{1}{2} v_x g_j - h_j.
$$ 
We have:
\begin{align*}
& \D_i \left(L_{1j}^{(g_j)}\right) = \frac{1}{2} v_{t_i x} g_j + \frac{1}{2} v_x (g_j)_{t_i} - \D_i (h_j),\\
& \D_{g_j} (L_{1i}) = \frac{1}{2} (g_j)_x v_{t_i} + \frac{1}{2} v_x (g_j)_{t_i} - \D_{g_j}(h_i).
\end{align*}
Upon using (\ref{lemma-Da}) and (\ref{eq def b}), and introducing the polynomial
$$
a_{ij}^{(g_j)} := g_j \var[1]{h_i}{v_x} + (g_j)_x \var[1]{h_i}{v_{xx}} + (g_j)_{xx} \var[1]{h_i}{v_{xxx}} + \ldots
$$
obtained from $a_{ij}$ through the replacement of $v_{t_j}$ by $g_j$, we find:
\begin{align*}
\D_i \left(L_{1j}^{(g_j)}\right) - \D_{g_j} (L_{1i}) 
&= \frac{1}{2} v_{t_i x} g_j - \frac{1}{2} (g_j)_x v_{t_i} - \D_i (h_j) + \D_{g_j}(h_i) \\
&= \frac{1}{2} v_{t_i x} g_j - \frac{1}{2} (g_j)_x v_{t_i} - \big( a_{ji} \big)_x + (g_j)_x v_{t_i} + \left( a_{ij}^{(g_j)} \right)_x - (g_i)_x g_j \\
&= \frac{1}{2} v_{t_i x} g_j + \frac{1}{2} (g_j)_x v_{t_i} - \left( a_{ji} - a_{ij}^{(g_j)} \right)_x - (g_i)_x g_j \\
&= \frac{1}{2} (v_{t_i}g_j)_x + \left( a_{ij}^{(g_j)} - a_{ji} \right)_x - (b_{ij})_x.
\end{align*}
We denote the antiderivative with respect to $x$ of this quantity by
$$L_{ij}^{(i)} := \frac{1}{2} v_{t_i}g_j + \left( a_{ij}^{(g_j)} - a_{ji} \right) - b_{ij}.$$
The analogous calculation with coordinates $x$ and $t_j$ yields
$$ \D_{g_i} (L_{1j}) - \D_j\left(L_{1i}^{(g_i)}\right) = -\frac{1}{2} (v_{t_j}
g_i)_x + \left( a_{ij} - a_{ji}^{(g_i)} \right)_x + (b_{ji})_x.$$
We denote its antiderivative by
$$L_{ij}^{(j)} := -\frac{1}{2} v_{t_j}
g_i + \left( a_{ij} - a_{ji}^{(g_i)} \right) + b_{ji}. $$

Now we look for a differential polynomial $L_{ij}[v]$ depending on the partial derivatives of $v$ with respect to $x$, $t_i$ and $t_j$ that reduces to $L_{ij}^{(i)}$ and to $L_{ij}^{(j)}$ after the substitutions $v_{t_j} = g_j$ and $v_{t_i} = g_i$, respectively. It turns out that there is a one-parameter family of such functions, given by 
$$ L_{ij} = c v_{t_i}v_{t_j} + (a_{ij} - a_{ji}) + \left( \tfrac{1}{2} - c  \right) v_{t_i} g_j - \left( \tfrac{1}{2} + c  \right) v_{t_j} g_i + \tfrac{1}{2}(b_{ji} - b_{ij}) + c g_i g_j$$
for $c \in \mathbb{R}$. Checking this is a straightforward calculation using Equation (\ref{lemma-b+b}). Our theory does not depend in any essential way on the choice of $L_{ij}$ within this family. For aesthetic reasons we chose $c=0$, which gives us Equation \eqref{Lij}.

\begin{remark}
We could also take $\cL$ to be the $c$-linear part of the form we have just obtained, i.e.\@ $\cL=\sum_{i<j}  (v_{t_i} - g_i)(v_{t_j} - g_j) \,\d t_i \wedge \d t_j$. One can think of this as choosing $c = \infty$. Such a two-form $\cL$ can be considered for any family of evolutionary equations $v_{t_i}=g_i[v]$. However,  due to the vanishing components $L_{1i}$, this form $\cL$ has no relation to the classical variational formulation of the individual differential equations $v_{xt_i}=(g_i)_x$.
\end{remark}

Eventually, Equation \eqref{varsym-ansatz} leads to the following closedness property.

\begin{prop}\label{prop-closedness}
The two-form $\cL = \sum_{i<j} L_{ij} \,\d t_i \wedge \d t_j$, with coefficients given by \eqref{L1i} and \eqref{Lij}, is closed as soon as $v$ solves all but one of the PKdV equations $v_{t_2} = g_2$, \ldots, $v_{t_N} = g_N$.
\end{prop}
\begin{proof}
We use the notation
\begin{equation}\label{eq dL}
\d \cL = \sum_{i < j < k} M_{ijk} \, \d t_i \wedge \d t_j \wedge \d t_k, \quad M_{ijk}=\D_k (L_{ij}) - \D_j (L_{ik}) + \D_i (L_{jk})
\end{equation}
We start by showing that $M_{1jk}=\D_k (L_{1j}) - \D_j (L_{1k}) + \D_{x} (L_{jk})$ vanishes as soon as either $v_{t_j} = g_j$ or $v_{t_k} = g_k$ is satisfied. Indeed, we have:
\begin{align}
M_{1jk}  & = \D_k  L_{1j} - \D_j L_{1k} + \D_x L_{jk} \notag\\
&= \frac{1}{2} v_{t_j t_k} v_x + \frac{1}{2} v_{t_j} v_{x t_k} - \D_k h_j - \frac{1}{2} v_{t_j t_k} v_x - \frac{1}{2} v_{t_k} v_{x t_j} + \D_j h_k \notag \\
&\qquad + \frac{1}{2} \big( v_{x t_j} g_k + v_{t_j} \D_x g_k - v_{x t_k} g_j - v_{t_k} \D_x g_j \big) \notag \\
&\qquad + \D_k h_j + v_{t_k} \D_x g_j - \D_j h_k - v_{t_j} \D_x g_k - \frac{1}{2}( g_k \D_x g_j - g_j \D_x g_k ) \notag \\
&= \frac{1}{2} \big(v_{t_j} v_{x t_k} - v_{t_k} v_{x t_j} + v_{x t_j} g_k - v_{t_j} \D_x g_k - v_{x t_k} g_j + v_{t_k} \D_x g_j - g_k \D_x g_j + g_j \D_x g_k \big) \notag\\
&= \frac{1}{2} (v_{t_j} - g_j)\D_x(v_{t_k} - g_k) - \frac{1}{2} (v_{t_k} - g_k) \D_x(v_{t_j} - g_j). \label{kdv-closed}
\end{align}
For the case $i,j,k > 1$, we assume without loss of generality that $v_{t_i} = g_i$ and $v_{t_j} = g_j$ are satisfied. We do not assume that $v_{t_k}=g_k$ holds, and correspondingly we do not make any identification involving $v_{t_k}$, $v_{x t_k}$, \ldots. Using Equation \eqref{kdv-closed}, we find:
\begin{align*}
\D_x M_{ijk}  & = \D_x  \left( \D_k (L_{ij}) - \D_j (L_{ik}) + \D_i (L_{jk}) \right)\\
&= \D_k \left( \D_i (L_{1j}) - \D_j (L_{1i}) \right) 
 - \D_j \left( \D_i (L_{1k}) - \D_k (L_{1i}) \right) 
+ \D_i \left( \D_j (L_{1k}) - \D_k (L_{1j}) \right) \\
&= 0.
\end{align*}
Since these polynomials do not contain constant terms, it follows that 
\[ \D_k (L_{ij}) - \D_j (L_{ik}) + \D_i (L_{jk}) = 0. \qedhere \]
\end{proof}

\begin{remark}
Assuming that the statement of Theorem \ref{thm-kdvEL} holds true, one can easily prove a somewhat weaker claim than Proposition \ref{prop-closedness}, namely that  the two-form $\cL$ is closed on simultaneous solutions of \emph{all} the PKdV equations. Indeed, by Proposition \ref{prop-dLconst}, $\d \cL$ is constant on solutions of the multi-time Euler-Lagrange equations $v_{t_i} = g_i$. Vanishing of this constant follows from the fact that $\d \cL=0$ on the trivial solution $v \equiv 0$.
\end{remark}

\subsection{The multi-time Euler-Lagrange equations}

\begin{proof}[Proof of Theorem \ref{thm-kdvEL}]
 We check all multi-time Euler-Lagrange equations \eqref{pEL-2-0}--\eqref{pEL-2-2} individually. If $N>3$, we fix $k > j > i >1$. If $N=3$, we take $j=3$, $i=2$, and in the following ignore all equations containing $k$. We use the convention $L_{ji} = - L_{ij}$, etc.

\subsubsection*{Equations (\ref{pEL-2-2})}

\begin{itemize}
\item The equations 
$$\var[1i]{L_{1i}}{v_{Ix t_i}} + \var[ij]{L_{ij}}{v_{It_i t_j}} + \var[1j]{L_{j1}}{v_{It_j x}} = 0$$
and
$$\var[ij]{L_{ij}}{v_{It_i t_j}} + \var[jk]{L_{jk}}{v_{It_j t_k}} + \var[ki]{L_{ki}}{v_{It_k t_i}}= 0$$
are trivial because all terms vanish.
\end{itemize}

\subsubsection*{Equations (\ref{pEL-2-1})}

\begin{itemize}
\item The equation
$$\var[1i]{L_{1i}}{v_{x}} = \var[ij]{L_{ji}}{v_{t_j}}$$
yields
\begin{align*}
\frac{1}{2} v_{t_i} - \var[1i]{h_i}{v_x} &= \frac{1}{2} g_i - \var[ij]{a_{ij}}{v_{t_j}} \\
&= \frac{1}{2} g_i  - \var[ij]{}{v_{t_j}} \left( v_{t_j} \var[1]{h_i}{v_x} +  v_{t_j x} \var[1]{h_i}{v_{xx}} +  v_{t_j xx} \var[1]{h_i}{v_{xxx}} +\ldots \right) \\
&= \frac{1}{2} g_i - \var[1]{h_i}{v_x}.
\end{align*}
This simplifies to the PKdV equation 
\begin{equation}\label{EL-EPKDdVi}
v_{t_i} = g_i.
\end{equation}

\item For $\alpha > 0$, the equation
$$\var[1i]{L_{1i}}{v_{x^{\alpha+1}}} = \var[ij]{L_{ji}}{v_{t_j x^\alpha}}$$
yields
\begin{align*}
-\var[1i]{h_i}{v_{x^{\alpha+1}}}
&= -\var[ij]{}{v_{t_j x^\alpha}} \left( v_{t_j} \var[1]{h_i}{v_x} +  v_{t_j x} \var[1]{h_i}{v_{xx}} 
+  v_{t_j xx} \var[1]{h_i}{v_{xxx}} 
+\ldots \right) \\
&= -\var[1]{h_i}{v_{x^{\alpha+1}}},
\end{align*}
which is trivial.

\item Similarly, the equation
$$\var[1j]{L_{1j}}{v_{x}} = \var[ij]{L_{ij}}{v_{t_i}}$$
yields PKdV equation
\begin{equation}\label{EL-EPKDdVj}
v_{t_j} = g_j,
\end{equation}
and for $\alpha > 0$, the equation
$$\var[1j]{L_{1j}}{v_{x^{\alpha+1}}} = \var[ij]{L_{ij}}{v_{t_i x^\alpha}}$$
is trivial.

\item All equations of the form 
$$\var[1i]{L_{1i}}{v_{x I}} = \var[ij]{L_{ji}}{v_{t_j I}} \quad (t_i \not\in I) \qquad \text{and} \qquad \var[1j]{L_{1j}}{v_{x I}} = \var[ij]{L_{ij}}{v_{t_i I}} \quad (t_j \not\in I)$$ 
where $I$ contains any $t_l$ ($l > 1$) are trivial because each term is zero.

\item The equations 
$$\var[1i]{L_{1i}}{v_{It_i}} = \var[1j]{L_{1j}}{v_{It_j}} \qquad (x \not\in I)$$
are trivial because both sides are zero for nonempty $I$ and both are equal to $\frac{1}{2} v_x$ for empty $I$.
\end{itemize}

\subsubsection*{ Equations (\ref{pEL-2-0})}

\begin{itemize}
\item By construction, the equations $\displaystyle \var[1i]{L_{1i}}{v} = 0$ for $i >1$ are the equations \begin{equation}\label{EL-PKDdV}
v_{x t_i} = \D_x g_i.
\end{equation}
For $I$ containing any $t_l$, $l > 1$, $l \neq i$, the equations $\displaystyle \var[1i]{L_{1i}}{v_{t_I}} = 0 $ are trivial.

\item The last family of equations we discuss as a lemma because its calculation is far from trivial.

\begin{lemma}\label{lemma-pELij}
The equations $\displaystyle \var[ij]{L_{ij}}{v_{x^\alpha}} = 0$ are corollaries of the PKdV equations.
\end{lemma}


%

\begin{proof}[Proof of Lemma \ref{lemma-pELij}]
From Equation \eqref{Lij} we see that the variational derivative of $L_{ij}$ contains only three nonzero terms,
\begin{equation}\label{varderij}
\var[ij]{L_{ij}}{v_{x^\alpha}} = \der{L_{ij}}{v_{x^\alpha}} - \D_{i} \left(\der{L_{ij}}{v_{x^\alpha t_i}} \right) - \D{j} \left(\der{L_{ij}}{v_{x^\alpha t_j}} \right). 
\end{equation}
To determine the first term we use an indirect method. Assume that the dimension of multi-time $N$ is at least 4 and fix $k > 1$ distinct from $i$ and $j$. Let $v$ be a solution of all PKdV equations except $v_{t_k} = g_k$. By Proposition \ref{prop-closedness} we have
\begin{equation}\label{kder}
\sum_I \der{L_{ij}}{v_I} v_{I t_k} = \D_{k} L_{ij} = \D_{j} L_{ik} - \D_{i} L_{jk}.
\end{equation}
Since $\der{L_{ij}}{v_I}$ does not contain any derivatives with respect to $t_k$, we can determine $\der{L_{ij}}{v_{x^\alpha}}$ by looking at the terms in the right hand side of Equation \eqref{kder} containing $v_{x^\alpha t_k}$. These are
\begin{align*}
&\quad \D_{j} \left( - \tfrac{1}{2} g_i v_{t_k} + v_{t_k} \var[1]{h_i}{v_x} + v_{x t_k} \var[1]{h_i}{v_{xx}} + \ldots \right) \\
&- \D_{i} \left( - \tfrac{1}{2} g_j v_{t_k} + v_{t_k} \var[1]{h_j}{v_x} + v_{x t_k} \var[1]{h_j}{v_{xx}} + \ldots \right).
\end{align*}
Now we expand the brackets. By again throwing out all terms that do not contain any $v_{x^\alpha t_k}$, and those that cancel modulo $v_{t_i} = g_i$ or $v_{t_j} = g_j$, we get
\begin{align*}
& - v_{t_k} \D_{j}\left( \var[1]{h_i}{v_x} \right) + v_{x t_k} \D_{j}\left( \var[1]{h_i}{v_{xx}} \right) + v_{xx t_k} \D_{j}\left( \var[1]{h_i}{v_{xxx}} \right) + \ldots  \\
& + v_{t_k} \D_{i}\left( \var[1]{h_j}{v_x} \right) - v_{x t_k} \D_{i}\left( \var[1]{h_j}{v_{xx}} \right) - v_{x t_k} \D_{i}\left( \var[1]{h_j}{v_{xxx}} \right) - \ldots .
\end{align*}
Comparing this to Equation \eqref{kder}, we find that
$$\der{L_{ij}}{v_{x^\alpha}} = -\D_{i}\left( \var[1]{h_j}{v_{x^{\alpha+1}}} \right) + \D_{j}\left( \var[1]{h_i}{v_{x^{\alpha+1}}} \right). $$

On the other hand we have
$$- \D_{i} \left(\der{L_{ij}}{v_{x^\alpha t_i}} \right) - \D_{j} \left(\der{L_{ij}}{v_{x^\alpha t_j}} \right)
= \D_{i} \left(\var[1]{h_j}{v_{x^{\alpha+1}}} \right) - \D_{j} \left(\var[1]{h_i}{v_{x^{\alpha+1}}} \right),$$
so Equation \eqref{varderij} implies that $\displaystyle \var[ij]{L_{ij}}{v_{x^\alpha}} = 0$ for any $\alpha$.

Since $\displaystyle \var[23]{L_{23}}{v_{x^\alpha}} = 0$ does not depend on the dimension $N \geq 3$, the result for $N \geq 4$ implies the claim for $N = 3$.
\end{proof}
\end{itemize}

\noindent This concludes the proof of Theorem \ref{thm-kdvEL}.
\end{proof}

\section{Relation to Hamiltonian formalism}\label{sect-Hamiltonian}

In this last section, we briefly discuss the connection between the closedness of $\cL$ and the involutivity of the corresponding Hamiltonians. 

In Proposition \ref{prop-dLconst} we saw that $\d \mathcal{L}$ is constant on solutions. For the one-dimensional case ($d=1$) with $\mathcal{L}$ depending on the first jet bundle only, it has been shown in \cite{S} that this is equivalent to the commutativity of the corresponding Hamiltonian flows. If the constant is zero then the Hamiltonians are in involution. Now we will prove a similar result for the two-dimensional case.

We will use a Poisson bracket on \emph{formal integrals}, i.e.\@ equivalence classes of functions modulo $x$-derivatives \cite[Chapter 1--2]{Dickey}. In this section, the integral sign $\int$ will always denote an equivalence class, not an integration operator. The Poisson bracket due to Gardner-Zakharov-Faddeev is defined by 
$$ \left\{ \tint F, \tint G \right\} 
= \int \left( \D_x\var[1]{F}{u} \right)\var[1]{G}{u}.$$
Using integration by parts, we see that this bracket is anti-symmetric. Less obvious is the fact that it satisfies the Jacobi identity \cite[Chapter 7]{Olver}. 
As we did when studying the KdV hierarchy, we introduce a potential $v$ that satisfies $v_x = u$, and we identify the space-coordinate $x$ with the first coordinate $t_1$ of multi-time. We can now re-write the Poisson bracket as
\begin{equation}\label{poisson}
\left\{ \tint F, \tint G \right\} 
= \int \left( \D_x \var[1]{F}{v_x} \right) \var[1]{G}{v_x} 
= - \int \var[1]{F}{v} \var[1]{G}{v_x},
\end{equation}
for functions $F$ and $G$ that depend on the $x$-derivatives of $v$ but not on $v$ itself.

Assume that the coefficients $L_{1j}$ of the Lagrangian two-from $\cL$ are given by
$$L_{1j} = \tfrac{1}{2} v_x v_{t_j} - h_j,$$
where $h_j$ is a differential polynomial in $v_x, v_{xx}, \ldots$. This is the case for the PKdV hierarchy. The $L_{1j}$ are Lagrangians of the equations 
$$v_{x t_j} = \D_x g_j \qquad \text{or} \qquad u_{t_j} = \D_x g_j,$$
where $g_j := \var[1]{h_j}{v_x}$, hence $\var[1]{h_j}{v} = - \D_x g_j$. It turns out that the formal integral $\tint h_j$ is the Hamilton functional for the equation $u_{t_j} = \D_x g_j$ with respect to the Poisson bracket \eqref{poisson}. Formally:
$$ \left\{ \tint h_j, u(y) \right\} 
 =\left\{ \tint h_j, \tint u \,\delta(\cdot - y) \right\} 
= -\int \var[1]{h_j}{v} \delta(x - y) \
= \D_x g_j (y),$$
where $\delta$ denotes the Dirac delta.

\begin{thm}
If $\d \cL = 0$ on solutions, then the Hamiltonians are in involution, 
$$\left\{ \tint h_i, \tint h_j \right\} = 0.$$
\end{thm}
\begin{proof} Recall notation (\ref{eq dL}). We have
\begin{align*}
\tint M_{1jk}  &= \int \big( \D_x L_{jk} - \D_j L_{1k} + \D_k L_{1j} \big)  \\
&= \int \big( - \D_j L_{1k} + \D_k L_{1j} \big)  \\
&= \int \left( -\frac{1}{2} v_{x t_j} v_{t_k} - \frac{1}{2} v_x v_{t_k t_j} + \D_j h_k + \frac{1}{2} v_{x t_k} v_{t_j} + \frac{1}{2} v_x v_{t_j t_k} - \D_k L_{1j} \right)  \\
&= \int \left( \frac{1}{2} (v_{x t_k} v_{t_j} - v_{x t_j} v_{t_k}) - \D_k L_{1j} + \D_j h_k \right) 
\end{align*}
Using Equation (\ref{eq def a}) (which, as opposed to Equation (\ref{eq def b}), is independent of the form of $h_i$ and $g_i$), the evolution equations $v_{t_j} = g_j$, and integration by parts, we find that
\begin{align*}
\tint M_{1jk} &= \int \left( \frac{1}{2} (v_{x t_k} v_{t_j} - v_{x t_j} v_{t_k}) - \D_x a_{jk} + v_{t_k} \D_x g_j + \D_x a_{kj} - v_{t_j} D_x g_k  \right)  \\
&= \int \left( -\frac{1}{2} (g_j \D_x g_k  - g_k \D_x g_j) - \D_x a_{jk} + \D_x a_{kj} \right)  \\
&= \int g_k \D_x g_j  \\
&= - \int \var[1]{h_j}{v} \var[1]{h_k}{v_x}  \\
&= \left\{ \tint h_j, \tint h_k \right\} .
\end{align*}
Hence if $\d \cL = 0$ on solutions of the evolution equations $v_{t_j} = g_j$, then the Hamilton functionals are in involution.
\end{proof}

\section{Conclusion}

We have formulated the pluri-Lagrangian theory of integrable hierarchies, and propose it as a definition of integrability. The motivation for this definition comes from the discrete case \cite{BPS,LN1,S} and the fact that we have established a relation with the Hamiltonian side of the theory. For the Hamiltonians to be in involution, we need the additional fact that the Lagrangian two-form is closed. However, we believe that the essential part of the theory is inherently contained in the pluri-Lagrangian formalism.

Since the KdV hierarchy is one of the most important examples of an integrable hierarchy, our construction of a pluri-Lagrangian structure for the PKdV hierarchy is an additional indication that the existence of a pluri-Lagrangian structure is a reasonable definition of integrability.

It is remarkable that multi-time Euler-Lagrange equations are capable of producing evolutionary equations. This is a striking difference from the discrete case, where the evolution equations (\emph{quad equations}) imply the multi-time Euler-Lagrange equations (\emph{corner equations}), but are themselves not variational \cite{BPS}.

\medskip

This research is supported by the Berlin Mathematical School and the DFG Collaborative Research Center TRR 109 ``Discretization in Geometry and Dynamics''.

\appendix

\section{A very short introduction to the variational bicomplex}
\label{appendix-bicomp}

Here we introduce the variational bicomplex and derive the basic results that we use in the text. We follow Dickey, who provides a more complete discussion in \cite[Chapter 19]{Dickey}. Another good source on a (subtly different) variational bicomplex is Anderson's unfinished manuscript \cite{Anderson}. For ease of notation we restrict to real fields $u: \mathbb{R}^N \rightarrow \mathbb{R}$, rather than vector-valued fields.

The space of $(p,q)$-forms $\cA^{(p,q)}$ consists of all formal sums
\[ \omega^{p,q} = \sum f \,\delta u_{I_1} \wedge \ldots \wedge \delta u_{I_p} \wedge \d t_{j_1} \wedge \ldots \wedge \d t_{j_q}, \]
where $f$ is a polynomial in $u$ and partial derivatives of $u$ of arbitrary order with respect to any coordinates. The vertical one-forms $\delta u_I$ are dual to the vector fields $\der{}{u_I}$. The action of the derivative $\D_i$ on $\omega^{p,q}$ is
\begin{align*}
\D_i \omega^{p,q}
&= \sum (\D_i f) \,\delta u_{I_1} \wedge \ldots \wedge \delta u_{I_p} \wedge \d t_{j_1} \wedge \ldots \wedge \d t_{j_q} \\
&\hspace{12mm} + f \,\delta u_{I_1i} \wedge \ldots \wedge \delta u_{I_p} \wedge \d t_{j_1} \wedge \ldots \wedge \d t_{j_q} \\
&\hspace{12mm} + \ldots + f \,\delta u_{I_1} \wedge \ldots \wedge \delta u_{I_pi} \wedge \d t_{j_1} \wedge \ldots \wedge \d t_{j_q}.
\end{align*}
The integral of $\omega^{p,q}$ over an $q$-dimensional manifold is the $(p,0)$-form defined by
\[ \int \omega^{p,q} = \sum \left(\int f \, \d t_{j_1} \wedge \ldots \wedge \d t_{j_q} \right) \delta u_{I_1} \wedge \ldots \wedge \delta u_{I_p}. \]
We call $(0,q)$-forms horizontal and $(p,0)$-forms vertical. The \emph{horizontal exterior derivative} $\d: \cA^{(p,q)} \rightarrow \cA^{(p,q+1)}$ and the \emph{vertical exterior derivative} $\delta: \cA^{(p,q)} \rightarrow \cA^{(p+1,q)}$ are defined by the anti-derivation property

\begin{enumerate}
\item[\rm a)]   $\displaystyle \d \left( \omega_1^{p_1,q_1} \wedge \omega_2^{p_2,q_2} \right) = \d \omega_1^{p_1,q_1} \wedge \omega_2^{p_2,q_2} + (-1)^{p_1+q_1} \, \omega_1^{p_1,q_1} \wedge \d \omega_2^{p_2,q_2}$,

\medskip
$\displaystyle \delta \left( \omega_1^{p_1,q_1} \wedge \omega_2^{p_2,q_2} \right) = \delta \omega_1^{p_1,q_1} \wedge \omega_2^{p_2,q_2} + (-1)^{p_1+q_1} \, \omega_1^{p_1,q_1} \wedge \delta \omega_2^{p_2,q_2}$,
\end{enumerate}
and by the way they act on $(0,0)$-, $(1,0)$-, and $(0,1)$-forms:
\begin{enumerate}
\item[\rm b)]   $\displaystyle \d f = \sum_j \D_j f \,\d t_j = \sum_j \left(\der{f}{t_j} + \sum_I \der{f}{u_I}u_{Ij} \right) \d t_j$, \quad
$\displaystyle \delta f = \sum_I \der{f}{u_I} \delta u_{I}$,

\item[\rm c)]   $\displaystyle \d (\delta u_I )= - \sum_j \delta u_{Ij} \wedge \d t_j,$ \quad 
$\delta(\delta u_I) = 0$,

\item[\rm d)]   $\d(\d x_j) = 0$, \quad 
$\delta (\d x_j) = 0$, \quad 
$\displaystyle \delta (\d u_I )= - \d (\delta u_I )= \sum_j \delta u_{Ij} \wedge \d t_j$.
\end{enumerate}
Properties a)--d) determine the action of $\d$ and $\delta$ on any form. The corresponding mapping diagram is known as the \emph{variational bicomplex}.
\begin{align*}
\begin{matrix}
\vdots			& 					& \vdots			& 					& 			& 					& \vdots 			& 					& \vdots 			\\
\uparrow \delta	& 				 	& \uparrow \delta	& 				 	& 			& 					& \uparrow \delta 	& 					& \uparrow \delta 	\\
\cA^{(1,0)}		& \xrightarrow{\d} 	& \cA^{(1,1)}		& \xrightarrow{\d} 	& \ldots	& \xrightarrow{\d}	& \cA^{(1,n-1)} 	& \xrightarrow{\d}	& \cA^{(1,n)} 		\\
\uparrow \delta	& 				 	& \uparrow \delta	& 				 	& 			& 					& \uparrow \delta 	& 					& \uparrow \delta 	\\
\cA^{(0,0)}		& \xrightarrow{\d} 	& \cA^{(0,1)}		& \xrightarrow{\d} 	& \ldots	& \xrightarrow{\d}	& \cA^{(0,n-1)}		& \xrightarrow{\d}	& \cA^{(0,n)}		\\
\end{matrix}
\end{align*}

The following claims follow immediately from the definitions.

\begin{prop}\label{prop-delta-d}
We have $\d^2 = \delta^2 = 0$ and $\d \delta + \delta \d = 0$.
\end{prop}

\begin{remark}
This implies that $\d + \delta: \cA^k \rightarrow \cA^{k+1}$, where $\cA^k := \bigcup_{i=0}^k \cA^{(i,k-i)}$, is an exterior derivative as well.
\end{remark}

\begin{prop}\label{prop-delta-D}
We have $\D_i \delta = \delta \D_i$.
\end{prop}

\begin{prop}\label{prop-delta-i}
For a differential polynomial $h$, define the corresponding vertical generalized vector field by $\partial_h := \sum_I h_I \der{}{u_I}$. We have $\d \i_{\partial_h} + \i_{\partial_h} \d = 0$.
\end{prop}
\begin{proof}
It suffices to show this for (0,0)-forms (polynomials $f$ in $u$ and partial derivatives of $u$), for (0,1)-forms $\d t_j$, and for (1,0)-forms $\delta u_I$. For (0,0)-forms, both terms of the claimed identity are zero:
\[ \d ( \i_{\partial_h} f ) = 0, \quad
\i_{\partial_h} ( \d f )= \i_{\partial_h} \bigg( \sum_j \D_j f \,\d t_j\bigg) = 0.
\]
Likewise for (0,1)-forms:
\[
\d ( \i_{\partial_h} \d t_j ) =  0, \quad \i_{\partial_h} ( \d \d t_j ) = 0. \]
For (1,0)-forms we find:
\[ \i_{\partial_h} ( \d \delta u_I )
= \i_{\partial_h} \bigg( - \sum_j \delta u_{Ij} \wedge \d t_j \bigg)
= - \sum_j h_{Ij} \, \d t_j
= - \d h_{I}
= - \d ( \i_{\partial_h} \delta u_I ). \qedhere \]
\end{proof}

\section{Proof of Lemma \ref{lemma-conv}}
\label{Appendix proof of lemma}

Assume that the action is stationary on all $d$-dimensional stepped surfaces in $\mathbb{R}^N$. Let $S$ be a smooth $d$-dimensional surface in $\mathbb{R}^N$. Partition the space $\mathbb{R}^N$ into hypercubes $C_i$ of edge length $\varepsilon$. We can choose this partitioning in such a way that the surface $S$ does not contain the center of any of the hypercubes. Denote $S_i^N := S \cap C_i$.

We give each hypercube its own coordinate system $[-1,1]^N \rightarrow C_i$ and identify the hypercube with its coordinates. In each \emph{punctured} hypercube $[-1,1]^N \setminus \{0\}$ we define a family of \emph{balloon maps}
$$ \mathcal{B}_\alpha^{N}: [-1,1]^N \setminus \{0\} \rightarrow [-1,1]^N \setminus \{0\}:
x \mapsto 
\begin{cases}
\dfrac{\alpha x}{\| x \|_\mathrm{max}} & \text{if } \| x \|_\mathrm{max} < \alpha \\
x & \text{if } \| x \|_\mathrm{max} \geq \alpha
\end{cases}$$
for $\alpha \in [0,1]$. Here, $ \| x \|_\mathrm{max} := \max(|x_1|, \ldots |x_N|)$ denotes the maximum norm with respect to the local coordinates. The idea is that from the center of each hypercube, we inflate a square balloon which pushes the curve away from the center, until it lies on the boundary of the hypercube. 

\begin{figure}[h]
\centering
\includegraphics[width=\linewidth]{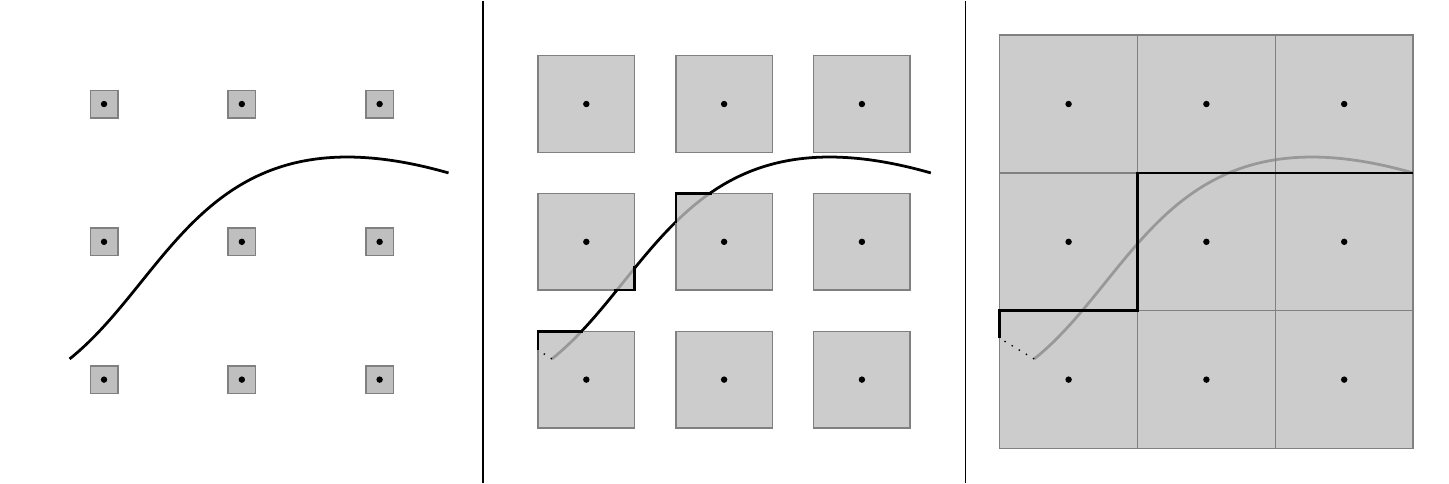}
\caption{Balloon maps in nine adjacent squares deforming a curve in $\mathbb{R}^2$. From left to right: $\alpha = 0.2$, $\alpha = 0.7$ and $\alpha = 1$.}
\label{Fig balloons}
\end{figure}

Indeed, the deformed curve $S_i^{N-1} := \mathcal{B}_1^{N} (S_i^N) = \mathcal{B}_1^{N} (S \cap C_i)$ lies on the boundary of the hypercube, i.e.\@ within the $(N-1)$-faces of the hypercube. We want it to lie within the $d$-faces of the hypercube, which would imply that it is a stepped surface. To achieve this, we introduce a balloon map 
$$ \mathcal{B}_\alpha^{N-1,j}: [-1,1]^{N-1} \setminus \{0\} \rightarrow [-1,1]^{N-1} \setminus \{0\}:
x \mapsto 
\begin{cases}
\dfrac{\alpha x}{\| x \|_\mathrm{max}} & \text{if } \| x \|_\mathrm{max} < \alpha \\
x & \text{if } \| x \|_\mathrm{max} \geq \alpha
\end{cases}$$
in each of the $(N-1)$-faces $C_i^j$ of the hypercube $C_i$, which pushes the surface into the $(N-2)$-faces. We denote the surface we obtain this way by $S_i^{N-2}$. If the surface happens to contain the center of a $(N-1)$-face, we can slightly perturb the surface without affecting the argument. By iterating this procedure, using balloon maps $\mathcal{B}_\alpha^{k,j}$ in each $k$-face $C_i^j$ ($N \geq k \geq d+1$), we obtain a surface $S_i^d$ that lies in the $d$-faces.

\begin{figure}[h]
\centering
\includegraphics[width=\linewidth]{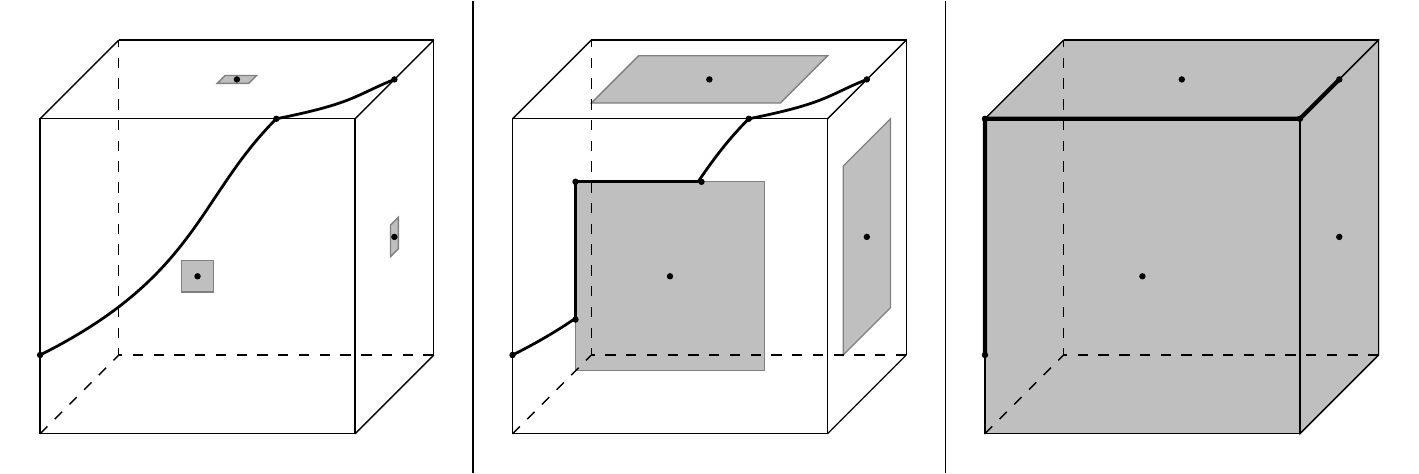}
\caption{The second and last iteration for a curve in $\mathbb{R}^3$. From left to right: $\alpha = 0.1$, $\alpha = 0.6$ and $\alpha = 1$.}
\end{figure}

Consider the $(d+1)$-dimensional surface 
$$ M_i := \bigcup_{k=d+1}^N \bigcup_{\substack{j:\  C_i^j \text{ is a} \\ \text{$k$-face of } C_i}} \bigcup_{\alpha \in [0,1]} \mathcal{B}_\alpha^{k,j} (S_i^k \cap C_i^j)$$
that is swept out by the consecutive application of the balloon maps to $S_i^N := S \cap C_i$. Assuming that $\varepsilon$ is small compared to the curvature of $S$, the $(d+1)$-dimensional volume of each of the $\bigcup_{\alpha \in [0,1]} \mathcal{B}_\alpha^{k,j} (S_i^k \cap C_i^j)$ is of the order $\varepsilon^{d+1}$. The number of such volumes making up $M_i$ only depends on the dimensions $N$ and $d$, not on $\varepsilon$, so the $(d+1)$-dimensional volume $|M_i|$ of $M_i$ is of the order $|M_i| = \mathcal{O}(\varepsilon^{d+1})$.

Now consider a variation $\cV$ with compact support and restrict the surface $S$ to this support. Denote by $\widehat{S} := \bigcup_i  S_i^d$ the stepped surface obtained from $S$ by repeated application of balloon maps in all the hypercubes, and by $M := \bigcup_i M_i$ the $(d+1)$-dimensional surface swept out by these balloon maps. The bounary of $M$ consists of $S$, $\widehat S$, and a small strip of area $\mathcal O(\varepsilon)$ connecting the boundaries of $S$ and $\widehat S$ (the dotted line in Figure \ref{Fig balloons}). The number of hypercubes intersecting $S$ is of order $\varepsilon^{-d}$, so $| M | = \mathcal{O}(\varepsilon^{-d}) {O}(\varepsilon^{d+1}) = \mathcal{O}(\varepsilon)$. It follows that
$$\left| \int_{\widehat{S}} \i_{\pr \cV} \delta \cL- \int_{S} \i_{\pr \cV} \delta \cL \right| 
= \left|\int_{\partial M} \i_{\pr \cV} \delta \cL \right| +\mathcal{O}(\varepsilon)
= \left|\int_{M} \d(\i_{\pr \cV} \delta \cL) \right| +\mathcal{O}(\varepsilon)
\rightarrow 0$$
as $\varepsilon \rightarrow 0$. By assumption, $\int_{\widehat{S}} \i_{\pr \cV} \delta \cL = 0$ for all $\varepsilon$, so the action on $S$ will be stationary as well. \qed

\bibliographystyle{amsalpha}

\end{document}